\documentstyle[11pt]{article}
\input amssymb.sty
\title{Physical properties as modal operators in the topos approach to  quantum mechanics}

\author{{\sc H. Freytes}\footnote{Fellow of the Consejo Nacional de Investigaciones Cient\'{\i}ficas y
T\'{e}cnicas (CONICET)}\ $^{1,2}$, \  {\sc G. Domenech}$^{4}$ \ and
\ {\sc C. de Ronde}$^{*}$$^{3,4,5}$}

\date{}

\begin{document}

\bibliographystyle{plain}

\maketitle

\begin{center}

\begin{small}
1. Dipartimento di Filosofia, Universit\`{a} di Cagliari \\ Viale
Merello 92, 09123, Cagliari - Italia.\\
2. Departamento de Matem\'{a}tica UNR-CONICET \\ Av. Pellegrini 250, CP 2000 Rosario - Argentina.\\
3. Instituto de Filosof\'{\i}a ``Dr. Alejandro Korn'' UBA-CONICET \\ Puan 480, Buenos Aires - Argentina.\\
4. Center Leo Apostel (CLEA)\\
5. Foundations of  the Exact Sciences (FUND) \\
Brussels Free University  Krijgskundestraat 33, 1160 Brussels -
Belgium
\end{small}
\end{center}

\begin{abstract}

\noindent In the framework of the topos approach to quantum
mechanics we give a representation of physical properties in terms
of modal operators on Heyting algebras. It allows us to introduce a
classical type study of the mentioned properties.
\end{abstract}

\begin{small}

\noindent {\em Keywords: Intuitionistic quantum logic, quantum phase
spaces, modal operators}

\noindent {\em PACS numbers: 02.10.De, 02.10.Ab}

\end{small}

\bibliography{pom}

\begin{thebibliography}{10}

\bibitem{aertsdaub2} D. Aerts and I. Daubechies, ``Mathematical condition for a sub-lattice of a propositional system  to represent a physical subsystem with a physical interpretation'', Lett. Math. Phys.
{\bf 3} (1979) 19-27.

\bibitem{aertsQL81} D. Aerts, ``Description of compound physical systems and logical interaction of physical systems'', in {\it Current Issues in Quantum Logic}, E. Beltrameti and B. van Fraassen (Editors) Plenum, New York, 1981, 381-405.

\bibitem{aertsrepmathphys84} D. Aerts, ``Construction of a structure which enables to describe the join system of a classical and a quantum system'', Rep. Math. Phys {\bf 20} (1984) 421-428.

\bibitem{aertsjmp84} D. Aerts, ``Construction of the tensor product of lattices of properties of physical entities'', J. Math. Phys. {\bf 25} (1984) 1434-1441.

\bibitem{AC08} S. Abramsky and B. Coecke, ``Categorical quantum
mechanics'', in \emph{Handbook of quantum logic and quantum
structures, vol. II}, Elsevier, 2008.

\bibitem{AB} S. Abramsky and A. Brandenburger, ``The sheaf-theoretic structure of
non-locality and contextuality'', New J. Phys. \textbf{13} (2011)
113036.

\bibitem{BIR} G. Birkhoff and J. von Neumann,  ``The logic of quantum mechanics'',  Ann. Math. {\bf 37} (1936) 823-843 .

\bibitem{BD} R. Balbes and Ph. Dwinger, \emph{Distributive Lattices}, University of Missouri Press, Columbia, 1974.

\bibitem{Bur} S. Burris,  H. P.  Sankappanavar, {\it A Course in Universal Algebra}, Graduate Text in Mathematics, Vol. 78. Springer-Verlag, New York Heidelberg Berlin, 1981.


\bibitem{chk13} B. Coecke, C. Heunen and A. Kissinger,
``Compositional quantum logic'', arXiv:1302.4900.

\bibitem{daCostadeRonde} N. da Costa and C. de Ronde, ``The
paraconsisten logic of quantum superpositions'', Found. Phys.
\textbf{43} (2013) 845-858.

\bibitem{QReasoning} M.L. Dalla Chiara, R. Giuntini and R. Greechhie,
\emph{Reasoning in quantum theory}, Kluwer, Dordrecht, 2004.

\bibitem{RFD14} C. de Ronde, H., Freytes and G. Domenech, ``Interpreting the Modal Kochen-Specker Theorem: Possibility and Many Worlds in Quantum Mechanics'', Stud. Hist. Phil. Mod. Phys. {\bf 45} (2014) 11-18.

\bibitem{contextual} G. Domenech and H. Freytes, ``Contextual logic for quantum systems'', J. Math. Phys. \textbf{ 46} (2005) 012102.

\bibitem{contextual2} G. Domenech, H. Freytes and C. de Ronde, ``Modal type othomodular logic''.
Math. Logic Quarterly \textbf{ 55} (2009) 287-299.

\bibitem{DHM} G. Domenech, F. Holik and C. Massri, ``A quantum logical and geometrical approach to the study of improper mixtures'', J. Math. Phys. \textbf{51} (2010) 052108.

\bibitem{DI2} A. D\"{o}ring and C.J. Isham, ``A topos foundation for theories of physics: II. Daseinisation and the liberation of quantum theory'', J. Math. Phys. {\bf 49}, (2008) 053516.

\bibitem{DI} A. D\"{o}ring and C.J. Isham, ``What is a thing?: topos theory in the foundations of physics'' in  \emph{New Structures for Physics}, Lecture Notes for Physics {\bf 813}, B. Coecke (Editor), Springer, Berlin-Heidelberg, 2010, 753-937.

\bibitem{dvu95} A. Dvure\u{c}enskij, ``Tensor product of difference posets and effect algebras'',  Int. J. Theor. Phys. \textbf{34} (1995) 1337-1348.

\bibitem{FREYD} P. J. Freyd, ``Aspects of Topoi''. Bull. Austral. Math. Soc. 7 (1972) 1-76.

\bibitem{GEL} I.M. Gelfand and M.A. Naimark, ``On the imbedding of normed rings into the ring of operators in Hilbert space'', Mat. Sbornik \textbf{12} (1943) 197-213.

\bibitem{gudderlibro78} S. P. Gudder, ``Some unresolved problems in quantum logic'', in {\it Mathematical Foundations of Quantum Theory}, A. R. Marlow (Editor), Academic, New York, 1978.

\bibitem{HLSW} C. Heunen, N. Landsman, B. Spitters and S. Wolters, ``The Gelfand spectrum of a noncommutative C$^{*}$-algebra: a topos theoretic approach'', J. Australian Math. Soc. \textbf{90} (2011) 39-52.

\bibitem{HLS} C. Heunen, N. Landsman and B. Spitters, ``A topos for algebraic quantum theory'', Comm. Math. Phys. \textbf{291} (2009) 63-110.

\bibitem{HEY1} A. Heyting, ``Die formalen Regeln der intuitionistischen Logik". Die Preussische Akademie der Wissenschaften. Sitzungsberichte. Physikalische-Mathematische Klasse, 42-56 (1930).

\bibitem{HEY2} A. Heyting, ``Die formalen Regeln der intuitionistischen Mathematik II, III". Die Preussische Akademie der Wissenschaften. Sitzungsberichte.
Physikalische-Mathematische Klasse., 57-71, 158-169 (1930).

\bibitem{isham} C. Isham, ``Topos methods in the foundations of physics'', in
\emph{Deep beauty: understanding the quantum world through
mathematical innovation}, H. Halvorson (Editor), Cambridge
University Press, 2010, 187-206.

\bibitem{JOHN} P. T. Johnstone, {\it Stone Spaces}, Cambridge Studies in Advanced Mathematics 3, Cambridge University Press, Cambridge,
1982.

\bibitem{LAW} F.W. Lawvere, ``Quantifiers and Sheaves''. Actes Congres Intern. Math. (1970), Tome 1, 329-334.

\bibitem{MAC} D.S. Macnab, ``Modal operators on Heyting algebras'', Alg. Univ.
\textbf{12} (1981) 5-29

\bibitem{MM} F. Maeda and S. Maeda, \emph{Theory of symetric
lattices}, Springer-Verlag, Berlin, 1970.

\bibitem{pulmJMP85}  S. Pulmannov\'{a},  ``Tensor product of quantum logics'', J. Math. Phys. {\bf 26}
(1985) 1-5.

\bibitem{FR81} C. H. Randall and D. J. Foulis, ``Empirical logic and tensor products'', in {\it Interpretation and Foundations of Quantum Theory}, H. Neumann (Editor),
Bibliographisches Institut, Mannheim, 1981,  21-28.

\bibitem{ZK} E. Zafiris and V. Karakostas, ``A categorial semantic representation of quantum event structures'', Found. Phys. \textbf{43} (2013) 1090-1123.

\end{thebibliography}

\newtheorem{theo}{Theorem}[section]

\newtheorem{definition}[theo]{Definition}

\newtheorem{lem}[theo]{Lemma}

\newtheorem{prop}[theo]{Proposition}

\newtheorem{coro}[theo]{Corollary}

\newtheorem{exam}[theo]{Example}

\newtheorem{rema}[theo]{Remark}{\hspace*{4mm}}

\newtheorem{example}[theo]{Example}

\newcommand{\proof}{\noindent {\em Proof:\/}{\hspace*{4mm}}}

\newcommand{\qed}{\hfill$\Box$}

\section*{Introduction}

Quantum mechanics (QM) is unanimously recognized to be one of the
most successful physical theories ever, but parallel to this it is
also widely acknowledged that many aspects of what quantum theory
means remain unexplained and obscure. And, although quite some
aspects that originally were considered problems ---e.g., the
Schr\"odinger cat situation related to the linearity of the
Schr\"odinger equation which gives rise to the superposition
principle--- are nowadays resources of technical applications,
nevertheless there still is a lack of a semantics and a conceptual
language for QM that would explain what the theory is talking about.

In the last years several approaches using category theory have been
used to search for an adequate and rigorous language for quantum
systems. First, both from a neo-realist point of view \cite{DI, HLS,
isham} or not \cite{ZK}, there are attempts  that relate algebraic
QM to topos theory, either recognizing the associated intrinsic
intuitionistic logic or equipping the structure with an external
intuitionistic logic. In these approaches, the quantum analogue of
classical phase space is captured by the notion of frame; i.e., a
complete Heyting algebra. There are also other schemes related to category theory which
attempt to describe several aspects of QM. For example,
contextuality and non-locality may be modeled using the framework of
sheaf theory \cite{AB} while monoidal categories can be used for
representing processes \cite{AC08}. This approach also enables a
consistent description of compound systems \cite{chk13}, a deep
difficulty for standard quantum logic (QL).

In this work we expose some logical characteristics related to the
intuitionistic approach to quantum phase spaces when the topos
approach to QM is considered. Moreover, we provide a representation
of physical  properties  as modal operators in a Heyting structure.
This representation allows us to analyze the classical and quantum
aspects of properties in terms of logical consequences. The paper is
organized as follows. In Section $1$, we recall some notions about
Heyting algebras and frames. In Section $2$, we review some basic
facts about algebraic approaches to QM. In Section $3$, we introduce
the Heyting algebra associated to the phase space when the topos
approach to quantum systems is considered. Section $4$ is dedicated
to the study of a representation of properties as modal operators in
a Heyting algebra. In this framework, we define a type of classical
interpretation for these properties. This interpretation will
describe semantic aspects related to the intuitionistic logic
associated to the phase space of the system. Finally, in Section $5$
we provide the conclusions.

\section{Heyting algebras and frames}

We recall from \cite{BD,Bur} some notions of universal algebra and
Heyting algebras that will play an important role in what follows.
Let $A$ be a non-empty set and $n$ be a non-negative integer. An
{\it $n$-ary} operation on $A$ is a function of the form $f:A^n
\rightarrow A$. In this way, $n$ is said to be the {\it arity of
$f$}. A {\it type} of algebras is a (possible infinite) sequence of
natural numbers $\tau = \{n_1, n_2 \ldots  \}$. An {\it algebra of
type $\tau$} is a pair $\langle A, F \rangle$ where, $A$ is a
non-empty set and $F = \{f_{n_1}, f_{n_2} \ldots\}$ is a set of
operation on $A$ such that $arity(f_i) = n_i$. The operations in $F$
are called {\it $\tau$-operations}. Let $A$ and $B$ be two algebras
of the same type $\tau$. A function $\alpha: A \rightarrow B$ is
{\it $\tau$-homomorphism} iff it preserves the $\tau$-operations.
Let $A$ be an algebra and $X\subseteq A$. We denote by $G_A(X)$ the
subalgebra of $A$ generated by the set $X$, i.e., the  smallest
subalgebra of $A$ containing $X$.

\begin{prop}\label{SG}
Let $A,B$ algebras of type $\tau$. Let $X$ be a subset of $A$ and $f:A\rightarrow B$ be a $\tau$-homomorphism. Then,

\begin{enumerate}
\item
$f(G_A(X)) = G_B(f(X))$ {\rm \cite[Theorem: 6.6 ]{Bur}}.

\item
If $A = G_A(X)$ and $g:A\rightarrow B$ is a $\tau$-homomorphism such
that the restrictions $g_{/X}$,  $f_{/X}$ coincides. Then, $f=g$
{\rm \cite[Theorem: 6.2 ]{Bur}}.
\end{enumerate}
\qed
\end{prop}

Let ${\cal A}$ be a category of algebras of type $\tau$ whose arrows
are the $\tau$-homomorphisms between algebras of ${\cal A}$. ${\cal
A}$ is a {\it variety} iff its objects form a class defined by
equations. It is well known that, if ${\cal A}$ is a variety, then
monomorphisms in ${\cal A}$ are exactly injective
$\tau$-homomorphisms.

An algebra $A \in {\cal A}$ is {\it injective in ${\cal A}$} iff,
for every monomorphism $f:B\rightarrow C$ and every homomorphism
$g:B \rightarrow A$, there exists a homomorphism $h:C\rightarrow A$
such that $g= hf$.

We shall focus our interest in two varieties, the variety of bounded
distributive lattices and the variety of Heyting algebras. The
following result will be used in Section 5.
\begin{theo}\label{INJ} {\rm \cite[V.9.-Theorem: 3]{BD}}
An algebra is injective in the variety of bounded distributive lattices iff it is a complete Boolean algebra.
\qed
\end{theo}

Heyting algebras provide an algebraic semantics for the
intuitionistic propositional calculus presented by Heyting in his
1930's papers \cite{HEY1, HEY2}. A {\it Heyting algebra}  \cite{BD}
is an algebra $ \langle A, \lor, \land, \rightarrow,  0  \rangle$ of
type $ \langle 2, 2,2,0  \rangle$ satisfying the following
equations:

\begin{enumerate}
\item[H1]
$ \langle A, \lor, \land, 0  \rangle$ is a lattice with universal lower bound $0$,

\item[H2]
$x\land y = x\land (x\rightarrow y)$,

\item[H3]
$x\land (y \rightarrow z) = x\land ((x\land y) \rightarrow (x\land z))$,

\item[H4]
$z\land ((x\land y) \rightarrow x) = z$.

\end{enumerate}

\noindent We denote by ${\mathbb{H}}$ the variety of Heyting
algebras. In agreement with the usual Heyting algebraic operations,
we define the negation $\neg x = x\rightarrow 0$ and $1 = \neg 0 $.

In each Heyting algebra $A$, the reduct $ \langle A, \lor, \land, 0,
1  \rangle$ is a bounded distributive lattice. The lattice order,
expressed in terms of the operation $\rightarrow$, is equivalent to
$a\leq b$ iff $1 = a\rightarrow b$. Moreover, for $a,b \in A$,
$a\rightarrow b = \bigvee\{x\in A: x\land a \leq b\}$. Boolean
algebras are Heyting algebras satisfying the equation $x\lor \neg x
= 1$. In this case, the operation $\rightarrow$ satisfies that
$x\rightarrow y = \neg x \lor y$.

Let $A$ be a Heyting algebra and $x\in A$. $x$ is said to be {\it
regular} iff $\neg \neg x = x$. We denote by $Reg(A)$ the set of
regular elements of $A$. $Reg(A)$ is a Boolean algebra under the
operations $x \lor_{R} y = \neg \neg (x \lor y)$, $x \land_{R} y = x
\land y$ and $x\rightarrow_{R} y = \neg \neg(\neg x \lor y)$ {\rm
\cite[IX.5.-Theorem: 3]{BD}}. In general $Reg(A)$ is not a
subalgebra of $A$. $x$ is said to be {\it central} iff $x\lor \neg x
= 1$. The set $Z(A)$ of central elements of $A$ constitutes a
Boolean subalgebra $A$. Note that $Z(A) \subseteq Reg(A)$. In
particular $Z(A) = Reg(A)$ iff the equation $\neg(x\land y) = \neg x
\lor \neg y$ is satisfied in $A$. The following result is well
known:

\begin{prop}\label{BOOLEQ}
Let $A$ be a Heyting algebra, then $A$ is a Boolean algebra iff $A =
Reg(A)$. \qed
\end{prop}

We have a special interest in the class of complete Heyting
algebras, i.e., Heyting algebras that are complete when considered as
lattices. Complete Heyting algebras are the objects of different
categories. They differ in their arrows, and thus get distinct
names. One of them is the category ${\it frames}$; i.e., complete
Heyting algebras whose arrows, called {\it frame homomorphisms}, are
functions preserving all joins, all finite meets, $0$ and $1$. The
Heyting operation $\rightarrow$ is not generally preserved by frame
homomorphisms. We also note that frames are very important
structures since they allow to study topological spaces in terms of
their open-sets lattices \cite{JOHN}.

\section{Algebraic approaches to quantum mechanics}

In classical physics every system can be described by specifying its
actual properties. Mathematically, this happens by representing the
state of the system  by a point $(p,\ q)$ in the corresponding phase
space $\Gamma$ and its properties by subsets of $\Gamma$, with a
structure of operations compatible with the usual mathematics of set
theory. Consequently, the propositional structure associated with
the properties of a classical system follows the rules of classical
logic. In the orthodox formulation of QM, a pure state of a system
is represented by a ray in the Hilbert space ${\mathcal H}$ and its
physical properties by closed subspaces of ${\mathcal H}$, which
with adequate definitions of meet and join operations give rise to
an orthomodular lattice \cite{MM}. This lattice, denoted by
$\mathcal{L}(\mathcal{H})$, is called the Hilbert lattice associated
to ${\mathcal H}$ and motivates the standard QL introduced in the
thirties by Birkhoff and von Neumann \cite{BIR}.

The traditional version of QL needs careful consideration for
several reasons. From an algebraic point of view, QL is founded on
the orthomodular lattice structure. But it is well known that the
variety of orthomodular lattices is strictly larger than the variety
generated by the Hilbert lattices. Thus, standard QL does not fully
capture the concept of the Hilbert lattice. From a physical point of
view there are several features which must be carefully considered:
if $P$ represents a proposition about the system, in general there
are superposition states in which it is wrong to say that either $P$
or its negation $\neg P$ hold in accordance with the association of
the join operation with the smallest closed subspace including the
projection represented by $P$ and its orthocomplement instead of
with their set theoretical union. However, the orthomodular
structure satisfies the equation $P\lor \neg P = 1$ which is a kind
of law of the excluded middle.  Thus, as discussed in
\cite{QReasoning}, it seems necessary to distinguish the logical law
of excluded middle from the semantic principle in which the truth of
the disjunction implies the truth of at least one of the
members.\footnote{For a discussion about contradiction and
superposition states see \cite{daCostadeRonde}.} Moreover, in spite
of the fact that the meet of its elements is well defined in the
lattice, there are conjunctions of (actual) properties that make no
sense because the corresponding operators do not commute. Thus, the
orthomodular structure shows a kind of conflict with the underlying
physical content of the theory. There is also a well known
difficulty with traditional forms of QL in relation to composite
systems, namely the lack of a canonical formalism for dealing with
the properties of the whole system when given the description of its
components. In fact, if $\mathcal{H}$$_{1}$ and $\mathcal{H}$$_{2}$
are the representatives of two systems, the postulates of QM say
that the tensor product
$\mathcal{H}=\mathcal{H}$$_{1}\otimes\mathcal{H}$$_{2}$ stands for
the representative of the composite. But the naive construction of
the lattice of propositions of the whole as the tensor product of
the lattices of the individuals  fails due to the lack of a product
of lattices \cite{aertsrepmathphys84,aertsjmp84,FR81}, or even
posets \cite{dvu95}. Mathematically, this is the expression of the
fact that the category of Hilbert lattices as objects and lattice
morphisms as arrows does not possess a categorial product due to the
failure of orthocomplementation
\cite{aertsdaub2,aertsQL81,gudderlibro78}. Attempts to vary the
conditions that define the product of lattices have been made
\cite{pulmJMP85}, but in all cases it results that the Hilbert
lattice factorizes only in the case in which one of the factors is a
Boolean lattice or when systems have never interacted, rendering the
construction either trivial or physically useless. For a complete
review, see \cite{dvu95}.\footnote{For the construction of a lattice
using convex sets instead of rays as states, see \cite{DHM}.} In
view of the mentioned characteristics of orthomodular systems of
propositions, there have been attempts to obtain ``more tractable''
structures (see for example \cite{AC08, chk13, DI, HLSW,
HLS}).\footnote{In this line, we have built in a previous paper a QL
that arises from considering a sheaf over a topological space
associated to the Boolean sublattices of the ortholattice of closed
subspaces of ${\mathcal H}$ \cite{contextual}. To do so, we defined
a valuation that respects contextuality (first translating the
Kochen-Specker (KS) theorem to topological terms \cite[Theorem
4.3]{contextual}) and a frame for the Kripke model of the language.
As frames are complete Heyting algebras, the resulting logic is an
intuitionistic one ---with restrictions on the allowed valuations
arising from the KS theorem---, thus it has ``good'' properties as
the distributive lattice structure and a nice definition of the
implication as a residue of the conjunction.}

The algebraic formulation of QM usually starts with the $C^*$-algebra of observables. This is
a complex algebra $A$ that is complete in a norm $\vert\vert \cdot
\vert\vert$ satisfying $\vert\vert xy \vert\vert \leq \vert\vert x
\vert\vert \vert\vert y \vert\vert$ and has an unary involutive
operation $^*$ such that $\vert\vert x^*x  \vert\vert =\vert\vert x
\vert\vert^2$. In this way, a quantum system is mathematically
modeled by a $C^*$-algebra. If $\mathcal{H}$ is a Hilbert space, the
algebra $B({\mathcal{H}})$ of all bounded operators of
$\mathcal{H}$, equipped with the usual norm and adjoint is an
example of $C^*$-algebra. By the Gelfand-Naimark theorem \cite{GEL},
any $C^*$-algebra is isomorphic to a norm-closed self-adjoint
subalgebra of $B({\mathcal{H}})$ for some Hilbert space
$\mathcal{H}$.

A von Neumann algebra $N$ is a special case of $C^*$-algebra $N
\subseteq B({\mathcal{H}})$ equal to its own bicommutant. More
precisely, if $N'$ is the set of all bounded operators on
$\mathcal{H}$ that commute with every element of $N$ then $N'' = N$.
Whereas $C^*$-algebra are usually considered in their norm-topology,
a von Neumann algebra  carries in addition a second interesting
topology, called the weak-topology, in which it is complete as well.
In this topology, one has convergence $x_n \rightarrow x$ iff, for
each density operator $\rho$, $tr\rho (x_n - x) \rightarrow 0$  in
$\mathcal{H}$ where $tr$ is the trace. A general $C^*$-algebra may
not have any nontrivial projections while a von Neumann algebra is
generated by its projections, i.e., elements satisfying $p^2 = p^* =
p$. In a von Neumann algebra, the projections are in natural
correspondence with the closed subspaces of a Hilbert space. In this
way, projections of a von Neumann algebra form a complete
orthomodular lattice. A state in a von Neumann algebra $N$ is a
linear functional $s:N \rightarrow \mathbb{C}$ that is continuous in
the weak topology and such that $s(x^* x) \geq 0$ and $s(1)=1$.

\section{Intuitionistic approach to phase spaces}

In the topos approach to QM \cite{DI, isham, HLS} it is possible to
encode physical properties in a Heyting algebra. This provides an
intuitionistic description for the phase space of the system. More
precisely, in a quantum system represented by a von Neumann algebra
$N$, the abelian subalgebras of $N$ represent {\it contexts} in
which, restricted to the context, the rules of classical logic hold (see for discussion \cite{RFD14}). Let $N$ be a von Neumann
algebra and ${\cal V}(N)$ be a family of commutative subalgebras of
$N$ which share the unit element with $N$. Consider the partial
ordered set $\langle {\cal V}(N), \subseteq \rangle$ viewed as the
small category whose arrows are defined by the partial order
$\subseteq$. In the topos approach the system is modeled in the
category of presheaves

$$\widehat{{\cal V}(N)} ={\bf Set}^{{{\cal V}(N)}^{op}}$$

\noindent Thus, the category $\widehat{{\cal V}(N)}$ can be seen as a category
of sets fibred over the contexts. Let $N$ be an abelian von Neumann
algebra. A {\it multiplicative state} is a state $s$ such that $s(x
y) = s(x)s(y)$. We denote by $\Sigma(N)$ the set of multiplicative
states in $N$ and the weak$^*$ topology is considered in
$\Sigma(N)$.  We recall that if a classical system is  modeled as an
abelian von Neumann algebra $N$,  $\Sigma(N)$ represents the phase
space of the system.

To model a quantum system, the {\it spectral presheaf} defined as the
functor: $$\Sigma: {{\cal V}(N)}^{op} \rightarrow {\bf Set}$$ such
that, ${\cal V}(N)\ni A \mapsto \Sigma(A)$ and, for each arrow
$f:A\rightarrow B$, (i.e., $f$ is the inclusion $A\subseteq B$),
$\Sigma(f)$ is the function $\Sigma(f): \Sigma(B) \rightarrow
\Sigma(A)$ such that $(\Sigma(f))(s) = s\vert_{A}$ is naturally
chosen as the state space.

Let $N$ be an abelian von Neumann algebra and ${\cal P}(N)$ be the
set of projections. Let $P \in {\cal P}(N)$. It is well known that
the set $C_P = \{s\in \Sigma(N): s(P)=1 \}$ is clopen  when the
weak$^*$ topology is considered in $\Sigma(N)$. Moreover, if we
consider the set $Clp(\Sigma(N))$ of all clopen sets in $\Sigma(N)$,
the function $C:{\cal P}(N) \rightarrow Clp(\Sigma(N))$ such that
$C(P) = C_P$ is a bijection.  A {\it clopen subobject} of the
spectral presheaf $\Sigma$ is a subfunctor $T$ of $\Sigma$ such that
for each $A \in {\cal V}(N)$, $T(A) \in Clp(\Sigma(N))$.

When considering $Sub_{cl}(\Sigma)$, the set of clopen subobjects of
$\Sigma$, we can see that,  $Sub_{cl}(\Sigma)$ is a bounded distributive
lattice where the operations $\lor, \land$ over clopen subobjects
are defined pointwise in each subalgebra of ${\cal V}(N)$, $0$ is
the empty subobject and $1= \Sigma$. In \cite[\S  2.3]{DI2} and  \cite[Theorem 2.5]{DI} the following result is proved:

\begin{theo}
$Sub_{cl}(\Sigma)$ is a complete Heyting algebra. \qed
\end{theo}

In a classical system, represented by a commutative von Neumann
algebra, the subsets of the phase space with usual set operations
define the logical (Boolean) structure of the system. For a quantum
system, represented by a von Neumann algebra whose phase space is
modeled by the spectral presheaf $\Sigma$, $Sub_{cl}(\Sigma)$
represents the logical structure of the system which is
intuitionistic. We will reefer to $Sub_{cl}(\Sigma)$ as the {\it the
algebra of propositions associated to the spectral presheaf
$\Sigma$}.

\section{Physical properties as modal operators}

In this section we study a class of classical interpretations for
quantum properties when the topos approach to quantum systems is
considered. For this purpose, we use the theory of modal operators
on Heyting algebras. In the orthodox approach, a classical
proposition is usually represented by a Boolean (also called
central) element of an orthomodular lattice \cite{contextual,
contextual2}. In particular, propositions about classical systems
are represented by a Boolean algebra. Suppose that $\mathcal{L}$ is
a lattice representing the propositional structure associated to a
quantum system. A classical interpretation of $\mathcal{L}$ implies
assuming that each $x\in \mathcal{L}$ has a classical complement
$\neg_c x$ and satisfies distributivity conditions in this
interpretation. Then, if $x$ is not a classical proposition in
$\mathcal{L}$, a classical interpretation of $\mathcal{L}$ must, al
least, endow a complement for $x$. Thus, a natural way to
algebraically represent classical interpretations are embeddings of
$\mathcal{L}$ into Boolean algebras, preserving lattice order
structure.

When the properties of a quantum system are encoded in $Sub_{cl}(\Sigma)$, we propose the following general formalization of the concept of classical interpretation for quantum properties:

\begin{definition}\label{LATINT}
{\rm Let $Sub_{cl}(\Sigma)$ be the algebra of propositions associated
to the spectral presheaf $\Sigma$. A {\it classical
interpretation} of the properties about the system is a lattice
embedding ${\cal C}:Sub_{cl}(\Sigma) \hookrightarrow B$ where $B$ is
a Boolean algebra.}
\end{definition}

\noindent Thus ${\cal C}:Sub_{cl}(\Sigma) \hookrightarrow B$
preserves $\lor, \land, 0,1$.  To study this type of classical
interpretation, we introduce the notion of logical consequence and
the notion of modal operator on a Heyting algebra. We note that the
theory of modal operators on Heyting algebras has its main
application in the theory of topoi and sheafication \cite{FREYD,
LAW}.

Let $A$ be a Heyting algebra and $a,b \in A$. We say that {\it $b$ is a logical consequence of $a$} iff $a\leq b$ or equivalently $1 =
a\rightarrow b$. We denote by $[a)$ the set of logical consequences of $a$. We remark that $[a)$ is the principal filter associated to
$a$ in $A$.

\begin{definition} \label{DEFMOD}
{\rm Let $A$ be a Heyting algebra. A modal operator on $A$
\cite{MAC} is a unary operation $j$ such that $$x\leq j(x),
\hspace{0.7cm} jj(x) = j(x), \hspace{0.7cm}  j(x\land y) = j(x)
\land j(y). $$ }
\end{definition}

Let $A$ be a Heyting algebra and $a\in A$. The operation
$\Diamond_a(x) = a\lor x$ defines a modal operator. Modal operators
of this form are known as {\it closed}. The operation
$\Diamond_{a\rightarrow}(x) = a\rightarrow x$ defines another modal
operator on $A$ and modal operators of this second form are known as
{\it open}.

\begin{prop}\label{CONS}
Let $A$ be a Heyting algebra and $a\in A$. Then:

\begin{enumerate}
\item
$Imag(\Diamond_a) = [a)$

\item
$\Diamond_{\neg a}(x) \leq \Diamond_{a\rightarrow}(x)$

\item
$Imag(\Diamond_{a\rightarrow}) \subseteq Imag(\Diamond_{\neg a}) = [\neg a) $

\item
$a$ is a Boolean element in $A$ iff $\Diamond_{\neg a} = \Diamond_{a\rightarrow}$

\end{enumerate}

\end{prop}

\begin{proof}
Let $x\in A$.

1) $x\in Imag(\Diamond_a)$ iff $x= a\lor t$ for some $t\in A$ iff
$x=a\lor x$ iff $x\in [a)$. \hspace{0.2cm} 2) Note that $a\land
(\neg a \lor x) = (a\land \neg a) \lor (a\land x) = a\land x \leq
x$. Then $\Diamond_{\neg a}(x) = \neg a \lor x \leq a\rightarrow x =
\Diamond_{a\rightarrow}(x)$. \hspace{0.2cm} 3) Since $a\land \neg a
\leq x$, $\neg a \leq a\rightarrow x$. Thus $a\rightarrow x \in
[\neg a)$ and $Imag(\Diamond_{a\rightarrow}) \subseteq
Imag(\Diamond_{\neg a}) = [\neg a) $. \hspace{0.2cm} 4) Suppose that
$a$ is a Boolean element, i.e., $\neg a \lor a =1$. On the one hand,
$x\leq a\rightarrow x$ and $\neg a \leq a\rightarrow x$. Then $\neg
a \lor x \leq a\rightarrow x$. On the other hand, suppose that
$t\land a \leq x$. Then $\neg a \lor x \geq \neg a \lor (t\land a )=
(\neg a \lor t) \land (\neg a \lor a) = (\neg a \lor t) \land 1 =
\neg a \lor t$. In particular $ \neg a \lor (a\rightarrow x) = \neg
a \lor \bigvee_{t\land a \leq x} t \leq \neg a \lor x$. Since $\neg
a \leq a\rightarrow x$, we have that $a\rightarrow x \leq \neg a
\lor x$. Hence $\Diamond_{a \rightarrow}(x) = a\rightarrow x = \neg
a \lor x = \Diamond_{\neg a} (x)$ and $\Diamond_{\neg a} =
\Diamond_{a\rightarrow}$. Now we suppose that $\Diamond_{\neg a} =
\Diamond_{a\rightarrow}$. Then $\neg a \lor a = \Diamond_{\neg a}(a)
= \Diamond_{a\rightarrow}(a) = a\rightarrow a = 1$. Hence $a$ is a
Boolean element in $A$. \qed
\end{proof}\\

The set $M(A)$ of all modal operators on $A$ is partially ordered by
the relation $j_1 \leq j_2$ iff $j_1(x) \leq j_2(x)$ for all $x\in
A$. If $A$ is a complete Heyting algebra, this partial order defines
a complete Heyting algebra structure on $M(A)$  \cite[Theorem
2.3]{MAC}  where $\bigwedge_i j_i$ is given by the operation
$(\bigwedge_i j_i)(x) = \bigwedge_i j_i(x)$. The implication
$j_1\rightarrow j_2$ is given by the operation $(j_1\rightarrow
j_2)(x) = \bigwedge \{j_1(y)\rightarrow j_2(y): y \geq x \}$. Joins
in $M(A)$ are defined as $j_1\lor j_2 = \bigwedge\{j\in M(A):
j_1,j_2 \leq j \} $.

\begin{theo}\label{EMB} {\rm \cite[$\S$ 2.6, $\S$ 2.7]{JOHN}}
Let $A$ be a complete Heyting algebra and $a\in A$ then:

\begin{enumerate}
\item
$\Diamond_a$ is a Boolean element in $M(A)$ and $\Diamond_{a\rightarrow}$ is its complement in $M(A)$.

\item
The map $a\mapsto \Diamond_a$ defines an injective frame homomorphism $A \rightarrow Reg(M(A))$.

\item
$a\mapsto \Diamond_a$ is an isomorphism iff $A$ is a Boolean
algebra.
\end{enumerate}\qed
\end{theo}

\noindent In general, $a\mapsto \Diamond_a$ does not preserve the operation
$\rightarrow$ except in the case in which $A$ is a Boolean algebra.

\begin{definition}
{\rm Let $A$ be a complete Heyting algebra. We define the algebra $A^{\Diamond}$ as the
Boolean subalgebra of $Reg(M(A))$ generated by $\{\Diamond_a,
\Diamond_{a\rightarrow}: a\in A\}$. }
\end{definition}

When considering the properties of the system encoded in $Sub_{cl}(\Sigma)$, the lattice embedding ${\cal C}_0 : Sub_{cl}(\Sigma) \rightarrow Sub_{cl}(\Sigma)^\Diamond$ such that ${\cal C}_0(a) = \Diamond_a$ can be seen as a classical interpretation of the quantum properties. We are interested in giving a meaning to this classical interpretation. To do so, we use the concept of logical consequence presented before Definition \ref{DEFMOD}.

Suppose that $a$ is a quantum property encoded in
$Sub_{cl}(\Sigma)$. Then, by Proposition \ref{CONS}-1, the classical
interpretation of $a$, given by the modal operator $\Diamond_a$, makes
reference to the logical consequences of $a$ in $Sub_{cl}(\Sigma)$.
The Boolean complement of $a$ in $Sub_{cl}(\Sigma)^\Diamond$ given
by $\Diamond_{a\rightarrow}$, by Theorem \ref{EMB}-1 and Proposition
\ref{CONS}-3, makes reference only to the consequences of  $\neg a$
in $Sub_{cl}(\Sigma)$ that have the form $a\rightarrow x$. Note
that, had $a$ been a property that commuted with all other
properties, $a$ would have been a  Boolean element in
$Sub_{cl}(\Sigma)$ and, by Proposition \ref{CONS}-4, the logical
consequences of $\neg a$ would have been of the form  $a\rightarrow
x$; i.e., the following identification would have hold: $\neg a
\approx \Diamond_{\neg a} = \Diamond_{a\rightarrow} $. This means
that thinking of $a$ as a classical property forces us to only
consider as the consequences of $\neg a$ those of the form
$a\rightarrow x$.

A first conclusion is that in the encoding of physical properties in
$Sub_{cl}(\Sigma)$, by Proposition \ref{CONS}-4, a classical
property is distinguished from a non classical one via the form of
the logical consequences of its negation in $Sub_{cl}(\Sigma)$. The
following example may help to make our assertion more clear:

\begin{example}
{\rm Suppose that $a,b \in Sub_{cl}(\Sigma)$ and $\Diamond_b \geq
\Diamond_{a\rightarrow}$. This means that the property $b$ is a
logical consequence of the complement of $a$ in the classical
interpretation $Sub_{cl}(\Sigma)^\Diamond$. Taking into account the
definition of $\Diamond_b$ and $\Diamond_{a\rightarrow}$, the
classical meaning of  $\Diamond_b \geq \Diamond_{a\rightarrow}$ is
that the logical consequences of $\neg a$ of the form $a\rightarrow
x$ have as logical consequence, the logical consequences of $b$ of
the form $b\lor x$. We remark the difference from the case in which
$a$ would have been a classical property. In this case ---in view of
Proposition \ref{CONS}-4, $\Diamond_{a\rightarrow} = \Diamond_{\neg
a}$ and  the meaning of $\Diamond_b \geq
\Diamond_{a\rightarrow}$---, if $p$ is a logical consequence of
$\neg a$ (i.e., $p$ is necessarily of the form $p= a\rightarrow x$
for some $x$) then $b\lor x$ is a logical consequence of $p$.
Clearly there exists a subtle difference between both
interpretations that could lead to contradictions when interpreting
$a$ classically without taking into account the distinction
$\Diamond_{\neg a}$ and $\Diamond_{a\rightarrow}$.}
\end{example}

Until now we have studied the natural meaning of the classical
interpretation ${\cal C}_0 : Sub_{cl}(\Sigma) \rightarrow
Sub_{cl}(\Sigma)^\Diamond$. But in fact ${\cal C}_0$ plays an
important role since it is present  in each possible classical
interpretation ${\cal C}: Sub_{cl}(\Sigma) \rightarrow B$ in the
sense of Definition \ref{LATINT}. The following theorem formally
describes this fact:

\begin{theo}\label{MAX}
Let $B$ be a Boolean algebra and $f:Sub_{cl}(\Sigma)\rightarrow B$
be a classical interpretation. Then there exists a unique injective
Boolean homomorphism $\widehat{f}: Sub_{cl}(\Sigma)^\Diamond
\rightarrow B$ such that the following diagram is commutative:

\begin{center}
\unitlength=1mm
\begin{picture}(20,20)(0,0)
\put(10,16){\vector(3,0){5}}
\put(2,10){\vector(0,-2){5}}
\put(8,4){\vector(1,1){7}}

\put(2,10){\makebox(13,0){$\equiv$}}

\put(2,16){\makebox(0,0){$Sub_{cl}(\Sigma)$}}
\put(20,16){\makebox(0,0){$B$}}
\put(2,0){\makebox(0,0){$Sub_{cl}(\Sigma)^\Diamond$}}
\put(4,20){\makebox(17,0){$f$}}
\put(2,8){\makebox(-6,0){${\cal C}$}}
\put(16,3){\makebox(-4,3){$\widehat{f}$}}
\end{picture}
\end{center}

\end{theo}

\begin{proof}
Let $f:Sub_{cl}(\Sigma)\rightarrow B$ be an injective lattice
homomorphism. Since $B$ is a Boolean algebra, $B$ can be embedded
into a complete Boolean algebra $B^*$. Thus we can see $B$ as a
Boolean subalgebra of $B^*$. By Theorem \ref{INJ}, $B^*$ is
injective in the variety of bounded distributive lattices. Then
there exists a bounded lattice homomorphism $\widehat{f}:
Sub_{cl}(\Sigma)\Diamond \rightarrow B^*$ such that the following
diagram is commutative:

\begin{center}
\unitlength=1mm
\begin{picture}(20,20)(0,0)
\put(10,16){\vector(3,0){5}}
\put(2,10){\vector(0,-2){5}}
\put(12,2){\vector(2,1){16}}

\put(4,8){\makebox(13,0){$\equiv$}}

\put(2,16){\makebox(0,0){$Sub_{cl}(\Sigma)$}}
\put(20,16){\makebox(0,0){$B$}}
\put(2,0){\makebox(0,0){$Sub_{cl}(\Sigma)^\Diamond$}}
\put(5,20){\makebox(17,0){$f$}}
\put(2,8){\makebox(-6,0){${\cal C}$}}
\put(25,1){\makebox(-4,3){$\widehat{f}$}}

\put(18,20){\makebox(17,0){$1_B$}}
\put(24,16){\vector(3,0){5}}
\put(34,16){\makebox(0,0){$B^*$}}

\end{picture}
\end{center}

We first prove that $\widehat{f}$ preserves complements. Let $x\in
Sub_{cl}(\Sigma)^\Diamond$. Since $1 = \widehat{f}(x\lor \neg x) =
\widehat{f}(x) \lor \widehat{f}(\neg x)$ and $0 = \widehat{f}(x\land
\neg x) = \widehat{f}(x) \land \widehat{f}(\neg x)$ then
$\widehat{f}(\neg x)$ is the complement of $\widehat{f}(x)$ in
$B^*$. Hence $\widehat{f}(\neg x) = \neg \widehat{f}(x)$. Thus
$\widehat{f}$ is a Boolean homomorphism. Now we prove that
$Imag(\widehat{f}) \subseteq B$. If $a\in Sub_{cl}(\Sigma)$ by the
commutativity of the diagram $f(a) = \widehat{f}(\Diamond_a) \in B$.
$\widehat{f}(\Diamond_{a \rightarrow})$ is the complement of
$\widehat{f}(\Diamond_a)$ in $B^*$ and $B$ is a Boolean subalgebra
of $B^*$ containing $f(\Diamond_a)$. Then the complement of
$f(\Diamond_a)$ in $B$ coincides with $\widehat{f}(\Diamond_{a
\rightarrow})$ because the complement is unique in a bounded distributive lattice. Thus $\widehat{f}(\Diamond_{a \rightarrow}) \in B$.
Note that $\{\Diamond_a, \Diamond_{a\rightarrow}: a\in Sub_{cl}(\Sigma)\}$ generates $Sub_{cl}(\Sigma)^\Diamond$. Then, by Proposition
\ref{SG}-1,
\begin{eqnarray*}
\widehat{f}(Sub_{cl}(\Sigma)^\Diamond)  & = & \widehat{f}(G_{Sub_{cl}(\Sigma)^\Diamond}\{\Diamond_a, \Diamond_{a\rightarrow}: a\in Sub_{cl}(\Sigma)\}) \\
& = & G_{B^*}(\{\widehat{f}(\Diamond_a), \widehat{f}(\Diamond_{a\rightarrow}): a\in Sub_{cl}(\Sigma)\})
\end{eqnarray*}
Since $\{\widehat{f}(\Diamond_a), \widehat{f}(\Diamond_{a\rightarrow}): a\in Sub_{cl}(\Sigma)\}
\subseteq B$ then the subalgebra of $B^*$ given by $G_{B^*}(\{\widehat{f}(\Diamond_a),
\widehat{f}(\Diamond_{a\rightarrow}): a\in Sub_{cl}(\Sigma)\})$  is a
Boolean subalgebra of $B$ and  $Imag(\widehat{f}) \subseteq B$. It
proves that

\begin{center}
\unitlength=1mm
\begin{picture}(20,20)(0,0)
\put(10,16){\vector(3,0){5}}
\put(2,10){\vector(0,-2){5}}
\put(8,4){\vector(1,1){7}}

\put(2,10){\makebox(13,0){$\equiv$}}

\put(2,16){\makebox(0,0){$Sub_{cl}(\Sigma)$}}
\put(20,16){\makebox(0,0){$B$}}
\put(2,0){\makebox(0,0){$Sub_{cl}(\Sigma)^\Diamond$}}
\put(4,20){\makebox(17,0){$f$}}
\put(2,8){\makebox(-6,0){${\cal C}$}}
\put(16,3){\makebox(-4,3){$\widehat{f}$}}
\end{picture}
\end{center}

\noindent By Proposition \ref{SG}-2, $\widehat{f}$ is the unique
Boolean homomorphism that makes commutative the diagram.
\qed
\end{proof}\\

The classical interpretation ${\cal C}_0 : Sub_{cl}(\Sigma)
\rightarrow Sub_{cl}(\Sigma)^\Diamond$ may be associated to a piece
of the classical language that describes some facts regarding the
logical consequences of the propositions about the system. Theorem
\ref{MAX} expresses the fact that any other classical interpretation
in the sense of Definition \ref{LATINT} only represents a  classical
enrichment of the language associated to
$Sub_{cl}(\Sigma)^\Diamond$. In other words, a classical
interpretation would represent an increasing of the expressive power
of the piece of the classical language associated to
$Sub_{cl}(\Sigma)^\Diamond$ that describes aspects of the logical
consequences of the propositions. Thus we may say that the classical
interpretations in the sense of Definition \ref{LATINT} only
describe semantic aspects of the logic of phase spaces.

\section{Conclusions}

The categorical approach to QM in the sense of \cite{DI, isham, HLS}
allows to establish classical interpretations of quantum properties.
We have rigorously described these interpretations in terms of modal
operators on Heyting algebras. When $a$ is a quantum property, its
classical interpretation given by the modal operator $\Diamond a$,
makes reference to the logical consequences of $a$ in
$Sub_{cl}(\Sigma)$. Its complement in $Sub_{cl}(\Sigma)^\Diamond$ is
given by $\Diamond_{a\rightarrow}$ and makes reference only to the
consequences of $\neg a$ that have the form $a\rightarrow x$.  Had
$a$ been a classical property, these would have been all the
consequences. But when $a$ is a genuine quantum property some of its
consequences are lacking.

\section*{Acknowledgements}

This work was partially supported by the following grants: VUB
project GOA67; FWO-research community W0.030.06; CONICET RES.
4541-12 (2013-2014) and  PIP 112-201101-00636, CONICET.

\end{document}